\documentclass[onecollarge]{svjour2} 

\smartqed  

\usepackage{amsmath,amssymb}
\usepackage{chngcntr}
\usepackage{mathtools}
\usepackage{graphicx}
\usepackage[all]{xy}

\usepackage{hyperref}

\usepackage{theorem}
\newtheorem{thm}{Theorem}

\newtheorem{prop}[thm]{Proposition}
{\theorembodyfont{\upshape}
\newtheorem{defin}[thm]{Definition}}

{\theorembodyfont{\upshape}
\newtheorem{exmp}[thm]{Example}}
{\theorembodyfont{\upshape}
\newtheorem{rem}[thm]{Remark}}

\numberwithin{thm}{subsection}
\numberwithin{equation}{section}

\journalname{Journal of Statistical Physics}

\begin{document}

\title{Microscopic Reversibility and Macroscopic Irreversibility: From the Viewpoint of Algorithmic Randomness}

\author{Ken Hiura \and
        Shin-ichi Sasa}

\institute{K. Hiura \and S. Sasa \at
              Department of Physics, Kyoto University, Kyoto 606-8502, Japan \\
              \email{hiura.ken.88n@st.kyoto-u.ac.jp, sasa@scphys.kyoto-u.ac.jp } }

\date{Received: date / Accepted: date}

\maketitle

\begin{abstract}
The emergence of deterministic and irreversible macroscopic behavior from deterministic and reversible microscopic dynamics is understood as a result of the law of large numbers. In this paper, we prove on the basis of the theory of algorithmic randomness that Martin-L\"of random initial microstates satisfy an irreversible macroscopic law in the Kac infinite chain model. We find that the time-reversed state of a random state is not random as well as it violates the macroscopic law.

\end{abstract}

\section{Introduction}
\label{sec:1}

\subsection{General introduction}


Let us consider a macroscopic fluid in an adiabatic container. A fundamental assumption of thermodynamics is that any thermodynamically isolated macroscopic system reaches a macroscopically stationary state, called the equilibrium state, after a sufficiently long time regardless of the initial state. The assumption is called the zeroth law of thermodynamics. The spatiotemporal change of macroscopic variables such as coarse-grained density fields of conserved quantities for simple fluids in this relaxation process is believed to be described universally by deterministic and irreversible hydrodynamic equations. This asymmetry in the direction of time is referred to as macroscopic irreversibility. In contrast, the macroscopic system microscopically consists of many interacting molecules. If the system is microscopically isolated, the time evolution of the constituent molecules is described by deterministic and reversible equations such as the classical Hamiltonian equations or quantum Schr{\"o}dinger equations. Therefore, the macroscopic irreversible laws should be formulated in microscopic reversible dynamical systems \cite{lebowitz,bricmont}.


A crucial concept in the formulation is the law of large numbers in the probability theory. Suppose that initial microscopic states are sampled from an initial probability measure. In general, the macroscopic behavior of each microscopic state may be quite different. However, if we choose an appropriate measure corresponding to a given nonequilibrium macrostate, for instance, the local Gibbs measure, as the initial probability measure, the validity of the deterministic macroscopic law is formulated as a result of the law of large numbers \cite{lps}. That is, there is a set of microscopic states satisfying the macroscopic equations with probability approaching one in a macroscopic limit. This is why the hydrodynamic equations describe even a single experimental result with high accuracy. Although to prove the law of large numbers for a given microscopic dynamics and initial probability measure is not an easy task in general, the above scenario is believed to be valid for a wide class of models and it is proved rigorously for specific models.


The law of large numbers gives a clear account of the emergence of macroscopic laws from microscopic dynamics. However, it refers only to the probability that the macroscopic law is satisfied and does not tell us which microscopic states among all realizable states obey the macroscopic law. When considering the reversibility paradox, one finds that this fact becomes problematic. Loschmidt pointed out that if a microscopic trajectory satisfying the microscopic equation of motion obeys the macroscopic law, the time-reversed trajectory is also a solution of the same equation due to the microscopic reversibility, but violates the macroscopic law due to the irreversibility of that law \cite{loschmidt}. Thus, the apparent inconsistency between the macroscopic irreversibility and the microscopic reversibility is relevant to individual trajectories. We note that the recurrence paradox posed by Zermelo also refers to a single trajectory \cite{zermelo}, but this paradox is resolved by considering the thermodynamic limit first. To resolve the reversibility paradox, it is desirable to have a more direct formulation studying individual microscopic states in the thermodynamic limit. In particular, we need a criterion to determine whether a given microscopic state belongs to a set characterized by typical macroscopic properties.


One possible approach to such a formulation is to use the theory of algorithmic randomness \cite{lv,nies,dh,gn}. This theory formalizes in an algorithmic manner the notion of typical sequences generated by a given stochastic process. Let us consider one-sided infinite binary sequences $x = x(0)x(1) \dots \in \{ 0, 1 \}^{\mathbb{N}}$ obtained by tossing a fair coin infinitely many times. Here, the underlying stochastic process is the uniform probability measure $\lambda$ on $\{ 0, 1 \}^{\mathbb{N}}$. Approximately speaking, a sequence $x \in \{ 0, 1 \}^{\mathbb{N}}$ is called \textit{random} if it satisfies no exceptional properties. For instance, a sequence in which $0$ and $1$ appear at a rate of $1/3$ and $2/3$, respectively, is not random under this experiment because its frequencies of $0$ and $1$ are exceptional. Exceptional properties are mathematically formalized by null sets. A set $N \subseteq \{ 0, 1 \}^{\mathbb{N}}$ is called a null set with respect to $\lambda$ if there is a sequence $(U_n)_{n \in \mathbb{N}}$ of open sets such that $N \subseteq \bigcap_n U_n$ and $\lambda(U_n) \leq 2^{-n}$. This consideration motivates the following definition of randomness: A sequence $x$ is random if $x \not\in N$ for any null set $N$. However, since any sequence $x$ is contained in a null set, for instance, a singleton $\{ x \}$, this definition turns out to be meaningless. We have to restrict the class of null sets to provide a meaningful definition of randomness. Following the celebrated idea of Martin-L{\"o}f \cite{martinlof}, we restrict null sets to ones that can be created \textit{feasibly in an algorithmic manner}. They specified such a feasible null set as one that is contained in $\bigcap_n U_n$ for a sequence $(U_n)_{n \in \mathbb{N}}$ of computably enumerable open sets in a uniform way with $\lambda(U_n) \leq 2^{-n}$. Then, a sequence $x \in \{ 0, 1 \}^{\mathbb{N}}$ is called Martin-L{\"o}f random with respect to $\lambda$ if $x \not\in N$ for any feasible null set $N$ in Martin-L{\"o}f's sense. See section \ref{sec:2} for details. Since the set of Martin-L{\"o}f random sequences has probability one, this definition captures some aspect of typicalness of sequences. Remarkably, Martin-L{\"o}f randomness, which characterizes random sequences as measure-theoretic typicalness, is equivalent to other notions of randomness characterized by incompressibility and unpredictability conditions. As a result, Martin-L{\"o}f randomness is regarded as a natural notion of randomness.


We apply the theory of algorithmic randomness to statistical physics on the basis of its characteristics that the algorithmically random sequences satisfy statistical properties such as the law of large numbers and of the iterated logarithm. The strong law of large numbers in the probability theory states that
\begin{align}
\label{eq:illn}
 \lim_{n \to \infty} \frac{1}{n} \sum_{i=0}^{n-1} x(i) = \frac{1}{2} \ \ \text{\textit{for almost all} $x$ with respect to} \ \lambda.
\end{align}
In contrast, the counterpart in the theory of algorithmic randomness states that
\begin{align}
\label{eq:ielln}
 \lim_{n \to \infty} \frac{1}{n} \sum_{i=0}^{n-1} x(i) = \frac{1}{2} \ \ \text{\textit{for random} $x$ with respect to} \ \lambda.
\end{align}
Although the former statement refers to only the probability that the law of large numbers is satisfied, the latter refers to individual sequences. In the context of statistical physics, the statistical properties of random sequences imply that the probability-theoretic statement,
\begin{quote}
 ``\textit{Almost all} microscopic states with respect to a probability measure obey a macroscopic law,''
\end{quote}
can be replaced by the pointwise one,
\begin{quote}
 ``\textit{Random} microscopic states with respect to a probability measure obey a macroscopic law.''
\end{quote}
Since this statement is expressed at the level of individual states, the notion of algorithmic randomness is useful to discuss the foundations of statistical mechanics beyond probability-theoretic statements.


The formulation with the notion of randomness provides a new perspective on the reversibility paradox. From a measure-theoretic point of view, the microscopic reversibility is consistent with our experience since the time-reversed state of a typical state with respect to a probability measure violates the macroscopic law, but has only an extremely small probability with respect to the same measure. In contrast, from a viewpoint of algorithmic randomness, the microscopic reversibility implies that the time-reversed state of a random state is not random because it is contained in a null set involved with a violation of the macroscopic law. This fact has an implication in the relation between the randomness of a physical state and the ease of preparation of that state. If we can specify a description of a binary sequence $x \in \{ 0, 1 \}^{\mathbb{N}}$ completely in an algorithmic manner, the sequence is not random because the singleton $\{ x \}$ is a null set in Martin-L\"of's sense. In contrast, a sequence we generate by a stochastic device such as tossing a coin many times is algorithmically random.  Now,  when we prepare a state of a physical system, we cannot avoid a certain source of noise. Although the relation between a stochastic device and noises in preparation of physical states is not obvious, we may say that it is difficult to experimentally prepare the time-reversed state of a random state, which is nonrandom, because we have to avoid all sources of noises to prepare a nonrandom state. In this manner, the theory of algorithmic randomness clarifies a conceptually new aspect of the reversibility paradox. We show a part of the results that the theory of algorithmic randomness reveals in this paper.

\subsection{Summary of results}
In this paper, we demonstrate with the aid of a pedagogical model how the emergence of macroscopic irreversible laws from reversible microscopic dynamics is formulated in terms of algorithmic randomness. We expect that the following results hold true for a wide class of models although we investigate a specific model in this paper.

We study a variant of the Kac ring model \cite{kac,go,mns}, which consists of two kinds of degrees of freedom, particles with spin $2x(i)-1 \in \{-1, 1\}$ ($i \in \mathbb{Z}$) on a one-dimensional infinite lattice $\mathbb{Z}$ and scatterers $y(i) \in \{ 0, 1\}$ located between particles. At each discrete time step, a particle at site $i \in \mathbb{Z}$ moves to site $i+1$. Then, the bit $x(i)$ is flipped if the scatterer at site $i$ is present, $y(i)=1$, and it remains its value if absent, $y(i)=0$. This evolution rule $\varphi: \{ 0, 1 \}^{\mathbb{Z}} \times \{ 0, 1 \}^{\mathbb{Z}} \to \{ 0, 1 \}^{\mathbb{Z}} \times \{ 0, 1 \}^{\mathbb{Z}}$ defines a discrete-time, deterministic and reversible dynamical system on $\{ 0, 1 \}^{\mathbb{Z}} \times \{ 0, 1 \}^{\mathbb{Z}}$. If we choose a set of macroscopic variables $m = (m_0, m_1)$ as the average magnetization $m_0^N$ and the fraction of scatterers $m_1^N$ over $2N+1$ sites around the origin, the system exhibits deterministic and irreversible behavior in the sense of the  law of large numbers: \textit{for almost all $(x,y)$ with respect to $\mu_{(1+m_0)/2} \times \mu_{m_1}$},
\begin{align}
 \lim_{N \to \infty} (m_i^N \circ \varphi^t)(x,y) = \Phi_i^t(m) \ \text{for all $i \in \{0, 1\}$ and $t \in \{0, \dots, T\}$},
\end{align}
where $\mu_p$ is the Bernoulli measure on $\{ 0, 1\}^{\mathbb{Z}}$ with parameter $p \in [0,1]$. $\varphi^t(x,y)$ is the microscopic state at time $t$ starting from a microstate $(x,y)$. $\Phi(m) = ((1-2m_1)m_0, m_1)$ represents a macroscopic law in the model. Our main claim of this paper is that algorithmically random microstates with respect to the initial probability measure satisfy the macroscopic law in the thermodynamic limit. That is to say, \textit{for Martin-L\"of random $(x,y)$ with respect to $\mu_{(1+m_0)/2} \times \mu_{m_1}$},
\begin{align}
 \lim_{N \to \infty} (m_i^N \circ \varphi^t)(x,y) = \Phi_i^t(m) \ \text{for all $i \in \{0, 1\}$ and $t \in \{0, \dots, T\}$}.
\end{align}
This result implies the zeroth law of thermodynamics for individual random microstates. Thus, the notion of algorithmic randomness opens the possibility of formulating macroscopic properties such as hydrodynamic equations and the zeroth law of thermodynamics at the level of individual microscopic states.

To quantify the irreversibility for individual trajectories, we define a quantity called the irreversible information loss as the logarithm of the ratio of probabilities at time $t$ of a microscopic state $\varphi^t(x,y)$ and the time-reversed one $(\pi \circ \varphi^t)(x,y)$ \cite{sk}. We prove that the irreversible information loss is positive for any random state, which implies the difficulty of realizing the time-reversed state of a random state in a measure-theoretic sense. The randomness notion sheds light on another aspect of the reversibility paradox. We show that the time-reversed state of a random microstate is not random as well as it violates the macroscopic law. This result is a manifestation of the macroscopic irreversibility that is revealed only after we apply the theory of algorithmic randomness to the problem in statistical physics.

\subsection{Previous studies}

There are a few works that have applied the theory of algorithmic randomness to statistical physics. The basic idea of such previous studies is to employ the Kolmogorov complexity for a microscopic state, which is the shortest program length outputting the state, and to present a formulation combining the Shannon entropy with the Kolmogorov complexity \cite{bennett}. For example, a new definition of entropy for microstates was proposed to provide some insight into Maxwell's demon problems \cite{zurek1,zurek2,caves}. It should be noted that the Kolmogorov complexity is independent of the probability measure, while the Martin-L{\"o}f randomness is defined for a probability measure. The most important relation between the two concepts is that an infinite sequence is Martin-L\"of random with respect to a probability measure $\mu$ if and only if the Kolmogorov complexity of the sequence is not smaller than the optimal compression length under the probability measure, $- \log \mu$, calculated from Shannon information theory. See Theorem \ref{thm:mlc} for the precise statement. Therefore, the difference of the optimal compression length from the Kolmogorov complexity, which is referred to as the randomness deficiency \cite{gn,gacs}, is the most important quantity to identify the Martin-L{\"o}f randomness. By using the randomness deficiency, we can express our statement as ``The randomness deficiency for an initial state diverges if the macroscopic behavior does not obey a macroscopic law.'' As far as we know, no statement using the Kolmogorov complexity of initial microstates has been addressed for describing the macroscopic irreversibility.

\subsection{Outline of the paper}
The remainder of the paper is organized as follows.

In section \ref{sec:2}, we review the theory of algorithmic randomness. To explain it in a self-contained manner, we include a brief review of computability theory and measure theory on the binary Cantor space. In section \ref{sec:3}, we first introduce a variant of the Kac ring model. We prove the law of large numbers in a measure-theoretic sense. With this in mind, we provide the pointwise version of the law of large numbers on the basis of algorithmic randomness. In section \ref{sec:4}, we define the Shannon and Boltzmann entropies. The pointwise law of large numbers leads to a pointwise version of the zeroth law of thermodynamics. In section \ref{sec:5}, we investigate the consequence of microscopic reversibility. We define a quantity called irreversible information loss quantifying the asymmetry between a microscopic trajectory and the time-reversed one, and prove the positivity of this quantity for random states. By using the reversibility property of microscopic dynamics, we construct a probability measure with respect to which the Boltzmann entropy decreases along the typical macroscopic trajectory. Similarly, we prove the nonrandomness of time-reversed states. In section \ref{sec:6}, we conclude with open problems and related topics.

\subsection{Notations}
We use the following notations throughout this paper.

$\mathbb{N}$, $\mathbb{Z}$, $\mathbb{Q}$, $\mathbb{R}$, $\mathbb{Q}_{\geq 0}$, and $\mathbb{R}_{\geq0}$ denote the set of natural numbers, integers, rational numbers, real numbers, nonnegative rational and real numbers, respectively. Let $\{ 0, 1 \}^{\mathbb{N}}$ denote the set of all infinite binary sequences, which is identified the set of all functions from $\mathbb{N}$ to $\{ 0, 1 \}$, $\{ 0, 1 \}^{<\mathbb{N}}$ the set of finite binary strings including the empty string $\square$, $|\sigma|$ the length of a string $\sigma \in \{ 0, 1 \}^{<\mathbb{N}}$, and $\sigma \tau$ the concatenation of finite string $\sigma$ and finite or infinite string $\tau$. A subset of natural numbers $A \subseteq \mathbb{N}$ is identified with its characteristic function $\chi_A \in \{ 0, 1 \}^{\mathbb{N}}$. For a finite string $\sigma$ and finite or infinite string $\tau$, we let $\sigma \sqsubseteq \tau$ denote that $\sigma$ is a prefix of $\tau$. For a finite or infinite string $x$, $x(n)$ denotes the $n$-th element of $x$ and $x \upharpoonright n$ or $x(0:n-1)$ its first $n$ bits $x(0)x(1)\dots x(n-1)$.

\section{Preliminaries}
\label{sec:2}
In this section, we review the algorithmic theory of randomness. Since this theory is based on computability theory, we also provide a brief review of computability theory. We hope that the paper will be read by theoretical physicists who are unfamiliar with computability theory. This section includes only a minimal set of concepts necessary for reading this paper and omits proofs of theorems. For more details of topics and proofs of theorems, see \cite{lv,nies,dh,gn} for the theory of algorithmic randomness and \cite{cooper,odifreddi1,odifreddi2} for the computability theory.

\subsection{Computability theory}
\label{subsec:21}
A function from a subset $A \subseteq \{ 0, 1 \}^{<\mathbb{N}}$ to $\{ 0, 1 \}^{<\mathbb{N}}$ is called a \textit{partial function} on $\{ 0, 1 \}^{<\mathbb{N}}$ and is denoted by $f: \subseteq \{ 0, 1 \}^{<\mathbb{N}} \to \{ 0, 1 \}^{<\mathbb{N}}$. The subset $A$ is called the domain of $f$ and is denoted by $\mathrm{dom}(f)$. The range of $f$ is denoted by $\mathrm{range}(f)$. If $A = \{ 0, 1 \}^{<\mathbb{N}}$, $f$ is called \textit{total} and is denoted by $f: \{ 0, 1 \}^{<\mathbb{N}} \to \{ 0, 1 \}^{<\mathbb{N}}$. A central concept of computability theory is the following.

\begin{defin}[computable function]
A partial function $f: \subseteq \{ 0, 1 \}^{<\mathbb{N}} \to \{ 0, 1 \}^{<\mathbb{N}}$ is \textit{computable} if there exists a Turing machine $M$ such that $M$ computes $f$. 
\end{defin}

Informally, each Turing machine $M$ represents a computer program and a partial function $f$ is computable if there is a program or algorithm such that for any input string $\sigma \in \{ 0, 1\}^{< \mathbb{N}}$, it either outputs $f(\sigma)$ if $f(\sigma)$ is defined, or it outputs nothing if $f(\sigma)$ is not defined. If we choose a coding function from $\{ 0, 1 \}^{<\mathbb{N}}$ to a countable object such as natural numbers, finite tuples of natural numbers, integers and rational numbers, we can extend the notion of computability of functions on $\{ 0, 1 \}^{<\mathbb{N}}$ to functions on the object. For instance, we can represent a natural number $n \in \mathbb{N}$ as a binary string $\beta(n) \in \{ 0, 1 \}^{<\mathbb{N}}$ by using the binary expansion. A partial function $f$ on $\mathbb{N}$ is called computable if there is a partial computable function $g: \subseteq \{ 0, 1 \}^{<\mathbb{N}} \to \{ 0, 1 \}^{<\mathbb{N}}$ with $g \circ \beta = \beta \circ f$. Similarly, a function $f: \subseteq \{ 0, 1 \}^{<\mathbb{N}} \to \mathbb{Q}$ is computable if there exists computable functions $\delta \times p \times q: \subseteq \{ 0, 1 \}^{<\mathbb{N}} \times \{ 0, 1 \}^{<\mathbb{N}} \times \{ 0, 1 \}^{<\mathbb{N}} \to \{ 0, 1\} \times \mathbb{N} \times \mathbb{N} \backslash \{ 0 \}$ such that $f(\sigma) = (-1)^{\delta(\sigma)} p(\sigma) / q(\sigma)$ for any $\sigma \in \{ 0, 1 \}^{<\mathbb{N}}$. All functions implemented in modern computers such as addition, multiplication, subtraction, division, and bounded summation are computable.

A set $A \subseteq \{ 0, 1 \}^{<\mathbb{N}}$ is computable if its characteristic function $\chi_A: \{ 0, 1 \}^{<\mathbb{N}} \to \{ 0, 1\}$ is computable. For instance, the set of all prime numbers is computable because there is an algorithm that decides whether a given natural number is a prime number or not. To formulate the notion of algorithmic randomness, we use a weaker notion of the computability of sets.

\begin{defin}[computably enumerable]
A set $A \subseteq \{ 0, 1 \}^{<\mathbb{N}}$ is \textit{computably enumerable} (\textit{c.e.}) if there exists a partial computable function $f: \subseteq \{ 0, 1 \}^{<\mathbb{N}} \to \{ 0, 1 \}^{<\mathbb{N}}$ such that $A=\mathrm{range}(f)$.
\end{defin}

This means that there exists an algorithm enumerating or listing all the members of the set. For example, for a polynomial $p(y_1,y_2)$ with integer coefficients, $D = \{ x \in \mathbb{N} : \exists y_1, y_2 \in \mathbb{N} \ p(y_1,y_2) = x \}$ may not be computable but is computably enumerable. It is easy to prove that $A \subseteq \{ 0, 1 \}^{<\mathbb{N}}$ is computable if and only if both $A$ and $\bar{A}$ are c.e. In particular, if $A$ is computable, then $A$ is c.e. Computable enumerability is a properly weaker notion than computability because there is a set that is computably enumerable but not computable. Examples of such sets are the halting problem of Turing machines and Hilbert's tenth problem.

We also define uniformly computable enumerability of sequences of sets.

\begin{defin}[uniformly c.e.]
A sequence $(A_n)_{n \in \mathbb{N}}$ of sets $A_n \subseteq \{ 0, 1 \}^{<\mathbb{N}}$ is \textit{computably enumerable uniformly in $n$} if there exists a partial computable function $f: \subseteq \{ 0, 1 \}^{<\mathbb{N}} \times \mathbb{N} \to \{ 0, 1 \}^{<\mathbb{N}}$ such that $A_n=\mathrm{range}(f(\cdot, n))$ for all $n \in \mathbb{N}$.
\end{defin}

A Turing machine is a special-purpose machine in the sense that the machine computes one computable function. Since we can code a program by a natural number in a computable manner, we can construct a \textit{universal Turing machine}, which is a model of present-day computers. This is why we can now implement any program by using only one computer.

\begin{thm}[universal Turing machine]
There is a partial computable function of two variables $g$ such that $g(e,x) = f_e(x)$ for any input $x$ and any partial computable function $f_e$ indexed by a natural number $e$.
\end{thm}

There is no coding that maps from finite strings to real numbers because the set of all real numbers is uncountable. Therefore, we say that a real number is computable if there exists a sequence of rationals approximating the real number from below and above in a computable way.

\begin{defin}[computable real, computable real-valued function]
A real number $x \in \mathbb{R}$ is \textit{lower semicomputable} if the set $\{ q \in \mathbb{Q} : q < x \}$ is computably enumerable. $x$ is \textit{upper semicomputable} if $-x$ is lower semicomputable. $x$ is \textit{computable} if it is both lower and upper semicomputable. Similarly, a real-valued function $f: \{ 0, 1 \}^{<\mathbb{N}} \to \mathbb{R}$ is \textit{lower semicomputable} if the set $\{ (\sigma, q) \in \{ 0, 1 \}^{<\mathbb{N}} \times \mathbb{Q} : q < f(\sigma) \}$ is computably enumerable. $f$ is \textit{upper semicomputable} if $-f$ is lower semicomputable. $f$ is \textit{computable} if it is both lower and upper semicomputable.
\end{defin}

\subsection{Topology and measure theory in Cantor space}

We review the basics of topology and measure theory on the set of infinite binary sequences. For a finite string $\sigma \in 2^{< \mathbb{N}}$, we use $[ \sigma ]$ to denote the cylinder set, that is, the set of all infinite binary sequences whose prefix is $\sigma$, $[ \sigma ] = \{ \sigma \tau : \tau \in \{ 0, 1 \}^{\mathbb{N}} \}$. For $S \subseteq \{ 0, 1 \}^{<\mathbb{N}}$, we let $[ S ] = \bigcup_{\sigma \in S} [ \sigma ]$.

\begin{defin}[c.e. open]

The Cantor space is $\{ 0, 1 \}^{\mathbb{N}}$ equipped with the product topology of a countable number of copies of the discrete topological space $\{ 0, 1\}$. The Cantor space has a countable basis of cylinder sets $\{ [ \sigma ] : \sigma \in \{ 0, 1 \}^{<\mathbb{N}} \}$. A subset  $A \subseteq \{ 0, 1 \}^{\mathbb{N}}$ is \textit{open} if it is the union of a subset of cylinder sets, that is,
\begin{align}
 A = [ S ] = \bigcup_{\sigma \in S} [ \sigma ]
\end{align}
for some subset of strings $S \subseteq \{ 0, 1\}^{<\mathbb{N}}$. If there exists a computably enumerable set $S$ such that $A = [ S ]$, then $A$ is called \textit{c.e. open}. A sequence $(A_n)_{n \in \mathbb{N}}$ of sets $A_n \subseteq \{ 0, 1 \}^{\mathbb{N}}$ is \textit{c.e. open uniformly in $n$} if there exists a sequence $(S_n)_{n \in \mathbb{N}}$ of c.e. sets uniformly in $n$ such that $A_n = [ S_n ]$ for all $n$.
\end{defin}

Let $(\{0, 1\}^{\mathbb{N}}, \mathcal{B})$ be the measurable space, where $\mathcal{B}$ is the Borel $\sigma$-algebra. It is known that a probability measure on $(\{0, 1\}^{\mathbb{N}}, \mathcal{B})$ can be constructed from a premeasure on $\{ [ \sigma ] : \sigma \in \{ 0, 1 \}^{<\mathbb{N}} \}$ with the aid of the Carath\'eodory's extension theorem. In the following, $\mu_{\rho}$ denotes the induced probability measure from a premeasure $\rho$ and is identified with the premeasure.

\begin{defin}[premeasure, computable measure]

A \textit{probability premeasure} is a function $\rho: \{ 0, 1 \}^{<\mathbb{N}} \to \mathbb{R}_{\geq 0}$ such that $\rho(\Box) = 1$ and $\rho(\sigma 0) + \rho(\sigma 1) = \rho(\sigma)$ for all $\sigma \in \{ 0, 1 \}^{<\mathbb{N}}$. The induced probability measure $\mu_{\rho}$ is \textit{computable} if $\rho$ is computable as a real-valued function.
\end{defin}

\begin{exmp}
\
\begin{itemize}
 \item[(1)] A premeasure $\rho(\sigma) = 2^{-|\sigma|}$ for $\sigma \in \{ 0, 1 \}^{<\mathbb{N}}$ induces the \textit{uniform measure} or the \textit{Lebesgue measure} $\lambda$.

 \item[(2)] Let $p$ be a real number such that $p \in (0,1)$. We set $\rho(1) = p$, $\rho(0) = 1-p$, and define a probability premeasure $\rho_{p} : \{ 0, 1 \}^{<\mathbb{N}} \to \mathbb{R}_{\geq 0}$ by
\begin{align}
 \rho_{p}(\sigma) = \prod_{i=0}^{|\sigma|-1} \rho(\sigma(i)).
\end{align}
We call the induced measure $\mu_{\rho_{p}}$ the \textit{Bernoulli measure} of parameter $\bar{p}$, which is denoted simply by $\mu_{p}$. The Bernoulli measure $\mu_{p}$ of parameter $p$ is computable if and only if $p$ is a computable real. We note that the Bernoulli measure with $p=1/2$ is the uniform measure $\lambda$.

\item[(3)] Let $x \in \{ 0, 1 \}^{\mathbb{N}}$ be a sequence. The \textit{Dirac measure} $\delta_x$ concentrated on $x$ is induced by the premeasure
\begin{align}
 \rho_x(\sigma) = \begin{cases}
 1 & \text{if} \ \sigma \sqsubseteq x \\
 0 & \text{otherwise}.
 \end{cases}
\end{align}
For any measurable set $A \subseteq \{ 0, 1 \}^{\mathbb{N}}$, 
\begin{align}
 \delta_x(A) = \begin{cases}
 1 & \text{if} \ x \in A \\
 0 & \text{if} \ x \not\in A.
 \end{cases}
\end{align}
\end{itemize}
\end{exmp}

We use the first Borel-Cantelli lemma to prove the strong form of the law of large numbers. Additionally, Martin-L\"of randomness has an alternative characterization in terms of the effective version of the Borel-Cantelli lemma (see Definition \ref{defin:solovay}).

\begin{thm}[first Borel-Cantelli lemma]
\label{thm:bc}
Let $(C_n)_{n \in \mathbb{N}}$ be a sequence of measurable sets. If $\sum_{n=0}^{\infty} \mu (C_n) < \infty$, then
\begin{align}
 \mu ( \{ x : x \in C_n \ \text{for infinitely many} \ n \} ) = \mu \left( \bigcap_{n=0}^{\infty} \bigcup_{k \geq n}^{\infty} C_k \right) = 0
\end{align}
\end{thm}

\subsection{Martin-L{\"o}f randomness}
\label{subsec:mlr}

In the probability theory, any realization of a stochastic process is assumed to occur randomly. In an $n$ times fair coin tossing experiment, a realization $0^n = 00 \dots 0$ ($n$ zeros) has the same probability $2^{-n}$ as any other realization. There is no difference between all realizations in this sense. Nevertheless, we believe that the relative frequencies of heads and tails approach asymptotically to $1/2$ as $n \to \infty$ under this experiment. This belief is represented by the strong law of large numbers in the probability theory. Although $0^{\mathbb{N}} = 00 \dots$ (infinitely many zeros) is a realizable outcome, it is not \textit{random} in that it does not satisfy the law of large numbers. Thus, it is possible to distinguish between random sequences and nonrandom ones according to the statistical laws that have the probability one. In other words, the notion of random sequences generated by a stochastic process is defined as ones having typical properties, or equivalently, having no exceptional properties. However, it is not clear what class of typical properties or exceptional properties we should choose to define random sequences. For instance, although $(01)^{\mathbb{N}} = 010101 \dots$ satisfies the law of large numbers, our intuition tells us that it is not a typical sequence generated by a fair coin tossing and therefore should not be random. Even if we require that the law of large numbers should hold for subsequences selected from a whole sequence by countable rules, there is a sequence satisfying the requirement but violating the law of the iterated logarithm, which is known as Ville's counterexample \cite{ville} (see also Theorem 6.5.1 in \cite{dh}). Hence, just the law of large numbers is not enough to characterize randomness. One naive idea is to consider all exceptional properties. We then define a set describing an exceptional property.

\begin{defin}[null set]
A set $N \subseteq \{ 0, 1 \}^{\mathbb{N}}$ is a \textit{null set} with respect to a probability measure $\mu$ if there is a sequence $(U_n)_{n \in \mathbb{N}}$ of open sets such that $N \subseteq \bigcap_{n \in \mathbb{N}} U_n$ and $\mu(U_n) \leq 2^{-n}$.
\end{defin}

\begin{exmp} \label{ex:null}
\
\begin{itemize}
 \item[(1)] For $x \in \{ 0, 1 \}^{\mathbb{N}}$, the one-element set $\{ x \}$ is a null set with respect to $\lambda$. Indeed, $\{ x \} = \bigcap_{n \in \mathbb{N}} U_n$, where $U_n = [x \upharpoonright n]$ with $\lambda (U_n) = 2^{-n}$.
 \item[(2)] $N = \{ x \in \{ 0, 1 \}^{\mathbb{N}} : x(2n)=1 \ \text{for all} \ n \in \mathbb{N} \}$ is a null set with respect to $\lambda$. Indeed, $N = \bigcap_{n \in \mathbb{N}} U_n$, where $U_n = \{ x \in \{ 0, 1 \}^{\mathbb{N}} : x(2k) = 1, \ 0 \leq k < n \}$ with $\lambda (U_n) = 2^{-n}$. 
\end{itemize}
\end{exmp}

Example \ref{ex:null} (1) shows that the naive idea fails because there is no sequence not contained in all null sets. To obtain a meaningful definition of random sequences, we have to restrict the class of null sets. The definition must satisfy the following two requirements at least.
\begin{enumerate}
 \item The set of random sequences is typical in measure-theoretic sense, that is, it has probability one.
 \item Sequences generated by some simple rule such as $(01)^{\mathbb{N}}$ are not random with respect to $\lambda$.
\end{enumerate}
We should note that a countable union of null sets is also a null set.

\begin{prop} \label{prop:null}
Let $( N_e )_{e \in \mathbb{N}}$ be a sequence of null sets with respect to a probability measure $\mu$. The countable union $\bigcup_{e \in \mathbb{N}} N_e$ of the sequence is  a null set.
\end{prop}

Therefore, if we choose a countable family of null sets to define random sequences, the first condition is automatically satisfied. We should also impose some computability conditions on the null sets if we interpret the generation by simple rules as listing elements in an algorithmic manner. The above argument motivates the following definition.

\begin{defin}[Martin-L{\"o}f random \cite{martinlof}]
Let $\mu$ be a computable probability measure on $\{ 0, 1 \}^{\mathbb{N}}$. A \textit{Martin-L{\"o}f test} with respect to the measure $\mu$ (ML $\mu$-test) is a sequence $( U_n)_{n \in \mathbb{N}}$ of c.e. open sets uniformly in $n$ such that $\mu(U_n) \leq 2^{-n}$ for all $n \in \mathbb{N}$. A set $N \subseteq \{ 0, 1 \}^{\mathbb{N}}$ is called a Martin-L{\"o}f null set with respect to $\mu$ (ML $\mu$-null set) if there is a Martin-L{\"o}f test $( U_n)_{n \in \mathbb{N}}$ such that $N \subseteq \bigcap_{n \in \mathbb{N}} U_n$. A sequence $x \in \{ 0, 1 \}^{\mathbb{N}}$ is Martin-L{\"o}f random with respect to $\mu$ (ML $\mu$-random) if $\{ x \}$ is not a Martin-L{\"o}f null set. $\textsf{MLR}_{\mu}$ denotes the set of ML $\mu$-random sequences.
\end{defin}

Martin-L\"of randomness satisfies the first requirement. Indeed, there are only countably many ML tests because there are only countably many c.e. sets. Since the union of all ML $\mu$-null sets is a null set with respect to $\mu$ from Proposition \ref{prop:null}, $\mu$-almost every sequence is ML $\mu$-random.

\begin{thm} \label{thm:full}
$\mu ( \sf{MLR}_{\mu} ) =1$.
\end{thm}

\begin{proof}
See Corollary 6.2.6 in \cite{dh}.
\end{proof}

To determine whether a given sequence is random or not, we have to implement a countable number of ML tests. However, the existence of a universal Turing machine implies that the union of all ML $\mu$-tests is also a ML $\mu$-test. Such a test is called universal Martin-L\"of $\mu$-test.

\begin{thm}[universal Martin-L\"of test]
There exists a Martin-L\"of test $\{ U_n \}_{n \in \mathbb{N}}$ with respect to $\mu$ such that for any ML $\mu$-test $\{ V_n \}_{n \in \mathbb{N}}$, $\bigcap_{n \in \mathbb{N}} V_n \subseteq \bigcap_{n \in \mathbb{N}} U_n$.
\end{thm}

\begin{proof}
See Fact 3.2.4 in \cite{nies} or Theorem 6.2.5 in \cite{dh}
\end{proof}

ML randomness also satisfies the second requirement. If $x \in \{ 0, 1 \}^{\mathbb{N}}$ is computable, then $x$ is not ML $\lambda$-random because $(U_n)_{n \in \mathbb{N}} = ([ x \upharpoonright n ])_{n \in \mathbb{N}}$ is a ML $\lambda$-test. We remark that even $z = 0^{\mathbb{N}}$ is ML random with respect to the Dirac measure $\delta_z$ concentrated on $z$.

The notion of randomness can be extended to objects in $\{0, 1\}^{\mathbb{Z}}$ and $\{ 0, 1 \}^{\mathbb{N}} \times \{ 0, 1 \}^{\mathbb{N}}$. We fix a bijective coding $\iota: \{0, 1\}^{\mathbb{Z}} \to \{ 0, 1 \}^{\mathbb{N}}$ in the following. We assign a two-sided infinite binary sequence $x = \dots x(-1)x(0)x(1) \dots \in \{0, 1\}^{\mathbb{Z}}$ to a one-sided infinite binary sequence
\begin{align}
  \iota(x) = x(0)x(-1)x(1)x(-2)x(2) \dots \in \{ 0, 1 \}^{\mathbb{N}}.
\end{align}
We let $\mu$ be a computable probability measure on $\{0, 1\}^{\mathbb{Z}}$. We say that $x \in \{0, 1\}^{\mathbb{Z}}$ is Martin-L\"of random with respect to $\mu$ if $\iota(x)$ is ML random with respect to $\mu \circ \iota^{-1}$. Similarly, we define a coding function from $\{ 0, 1 \}^{\mathbb{N}} \times \{ 0, 1 \}^{\mathbb{N}}$ to $\{ 0, 1 \}^{\mathbb{N}}$ as
\begin{align}
 \kappa(x,y) = x(0) y(0) x(1) y(1) \dots \in \{ 0, 1 \}^{\mathbb{N}} \ \text{for} \ (x,y) \in \{ 0, 1 \}^{\mathbb{N}} \times \{ 0, 1 \}^{\mathbb{N}}.
\end{align}
For a computable probability measure $\mu$ on $\{ 0, 1 \}^{\mathbb{N}} \times \{ 0, 1 \}^{\mathbb{N}}$, $(x,y) \in \{ 0, 1 \}^{\mathbb{N}} \times \{ 0, 1 \}^{\mathbb{N}}$ is called Martin-L\"of random with respect to $\mu$ if $\kappa(x,y)$ is ML random with respect to $\mu \circ \kappa^{-1}$.

There is an alternative characterization of Martin-L{\"o}f randomness in term of Solovay tests. We use this equivalence to prove the effective law of large numbers.

\begin{defin}[Solovay random] \label{defin:solovay}
Let $\mu$ be a computable probability measure. A \textit{Solovay test} is a sequence $(S_n)_{n \in \mathbb{N}}$ of c.e. open sets uniformly in $n$ such that $\sum_n \mu (S_n) < \infty$. $x \in \{ 0, 1 \}^{\mathbb{N}}$ is \textit{Solovay $\mu$-random} if $x$ is in only finitely many $S_n$.
\end{defin}

\begin{prop}
An element $x \in \{ 0, 1 \}^{\mathbb{N}}$ is Martin-L{\"o}f $\mu$-random iff it is Solovay $\mu$-random.
\end{prop}

\begin{proof}
See Proposition 3.2.19 in \cite{nies} or Theorem 6.2.8 in \cite{dh}.
\end{proof}

\subsection{Robustness of Martin-L{\"o}f randomness}
\label{subsec:rob}

The definition of Martin-L\"of randomness refers to an underlying probability measure. Therefore, even if a sequence $x \in \{ 0, 1 \}^{\mathbb{N}}$ is ML $\mu$-random, it may not be random with respect to another measure $\nu$. A trivial example is a sequence that is $\lambda$-random but not random with respect to the Bernoulli measure $\mu_{1/3}$. However, ML randomness is a robust notion in the sense that it is preserved by simple transformations. Here a simple transformation means a computable function \textit{from $\{ 0, 1 \}^{\mathbb{N}}$ to $\{ 0, 1 \}^{\mathbb{N}}$}, which is defined through a partial computable monotone function on $\{ 0, 1 \}^{<\mathbb{N}}$. 

\begin{defin}
A partial function $f: \subseteq \{ 0, 1 \}^{<\mathbb{N}} \to \{ 0, 1 \}^{<\mathbb{N}}$ is \textit{monotone} if
\begin{align}
  \sigma \sqsubseteq \tau \Rightarrow f(\sigma) \sqsubseteq f(\tau)
\end{align}
holds for all $\sigma, \tau \in \mathrm{dom}(f)$. For a partial monotone function $f: \subseteq \{ 0, 1 \}^{<\mathbb{N}} \to \{ 0, 1 \}^{<\mathbb{N}}$, we define a partial function $\hat{f}:\subseteq \{ 0, 1 \}^{\mathbb{N}} \to \{ 0, 1 \}^{\mathbb{N}}$ as
\begin{align}
 \hat{f}(x) = \begin{cases}
 \bigcup_{\sigma \sqsubseteq x} f(\sigma) & \text{if} \ \sup \{ |\sigma| : \sigma \sqsubseteq x, \sigma \in \mathrm{dom}(f) \} = \infty \\
 \text{undefined} & \text{otherwise}
 \end{cases}
\end{align}
for $x \in \{ 0, 1 \}^{\mathbb{N}}$. A partial function $F: \subseteq \{ 0, 1 \}^{\mathbb{N}} \to \{ 0, 1 \}^{\mathbb{N}}$ is \textit{computable} if there exists a partial computable monotone function $f: \subseteq \{ 0, 1 \}^{<\mathbb{N}} \to \{ 0, 1 \}^{<\mathbb{N}}$ with $F = \hat{f}$. If $\mathrm{dom}(F) = \{ 0, 1 \}^{\mathbb{N}}$, the function $F$ is called \textit{total}. Hereafter, we consider only total (computable) functions on $\{ 0, 1 \}^{\mathbb{N}}$.
\end{defin}

We define an image measure $\mu F^{-1}$ of $\mu$ under a function $F: \{ 0, 1 \}^{\mathbb{N}} \to \{ 0, 1 \}^{\mathbb{N}}$ as
\begin{align}
 (\mu F^{-1})(A) = \mu \left( F^{-1}(A) \right)
\end{align}
for any measurable set $A \in \mathcal{B}$. Then, computable functions on $\{ 0, 1 \}^{\mathbb{N}}$ preserve the computability of probability measures and the ML randomness.

\begin{prop}
Let $\mu$ be a measure on $\{ 0, 1 \}^{\mathbb{N}}$ and $F: \{ 0, 1 \}^{\mathbb{N}} \to \{ 0, 1 \}^{\mathbb{N}}$ be a total computable function. If $\mu$ is computable, $\mu F^{-1}$ is computable.
\end{prop}

\begin{proof}
See Lemma 2.6 in \cite{bp}.
\end{proof}

\begin{thm}[conservation of ML randomness]
\label{thm:cons}
Let $\mu$ be a computable probability measure on $\{ 0, 1 \}^{\mathbb{N}}$ and $F:  \{ 0, 1 \}^{\mathbb{N}} \to \{ 0, 1 \}^{\mathbb{N}}$ a total computable function. If $x \in \{ 0, 1 \}^{\mathbb{N}}$ is ML $\mu$-random, then $F(x)$ is ML $\mu F^{-1}$-random.
\end{thm}

\begin{proof}
See Theorem 3.2 in \cite{bp}.
\end{proof}

Martin-L\"of randomness has robustness in another sense. It has other characterizations in terms of incompressibility and unpredictability. Let us consider a sequence. If the sequence has a simple structure, it can be compressed into a shorter length one by using an algorithm. For instance, $(01)^{100000}$ can be transformed to the shorter program ``output one hundred thousand $01$s ''. Conversely, if the sequence is ``random'', there is no simple description of it. This consideration leads to the idea of the Kolmogorov complexity \cite{solomonoff1,solomonoff2,kolmogorov}. Here we use a prefix-free version of the Kolmogorov complexity for technical reasons.

\begin{defin}[prefix-free computable function]
A set of strings $A \subseteq \{ 0, 1\}^{< \mathbb{N}}$ is called \textit{prefix-free} if for any two distinct elements $\sigma$ and $\tau$ in $A$, $\sigma$ is not a prefix of $\tau$. A partial computable function $f: \subseteq \{ 0, 1\}^{< \mathbb{N}} \to \{ 0, 1\}^{< \mathbb{N}}$ is \textit{prefix-free} if $\mathrm{dom}(f)$ is prefix-free.
\end{defin}

The Kolmogorov complexity of a string $\sigma \in \{ 0, 1 \}^{<\mathbb{N}}$ with respect to a prefix-free computable function $f$ is defined as the length of a shortest program (string) $\tau \in \{ 0, 1 \}^{<\mathbb{N}}$ with $f(\tau) = \sigma$. That is,
\begin{align}
 K_f(\sigma) = \min \{ | \tau | : f(\tau) = \sigma \},
\end{align}
where the minimum is taken to be $\infty$ if the set after the ``$\min$'' is empty. The Kolmogorov complexity of a string $\sigma$ depends on the underlying prefix-free computable function $f$. However, there exists an optimal prefix-free computable function $U$ in the sense that if for any partial computable function $f: \subseteq \{ 0, 1 \}^{<\mathbb{N}} \to \{ 0, 1 \}^{<\mathbb{N}}$, there is a positive constant $c_f < \infty$ such that for all $\sigma \in \{ 0, 1 \}^{<\mathbb{N}}$,
\begin{align}
\label{eq:pop}
 K_U(\sigma) \leq K_f(\sigma) + c_f.
\end{align}
Thus, if a string can hardly be compressed by an optimal function, then the string cannot be compressed by any computable function. In other words, the Kolmogorov complexity is an intrinsic property of strings.

\begin{defin}[prefix-free Kolmogorov complexity]
We fix an optimal prefix-free computable function $U$ and define the \textit{prefix-free Kolmogorov complexity} $K(\sigma)$ of a string $\sigma \in \{ 0, 1 \}^{<\mathbb{N}}$ as $K(\sigma) = K_U(\sigma)$.
\end{defin}

According to the following theorem, the Kolmogorov complexity provides a characterization of randomness in terms of incompressibility. That is to say, the measure-theoretic typicalness of a sequence is equivalent to the incompressibility of it in the sense of the Kolmogorov complexity.

\begin{thm}[equivalence between ML randomness and complexity randomness]
\label{thm:mlc}
Let $\mu$ be a computable probability measure. A binary sequence $x \in \{ 0, 1 \}^{\mathbb{N}}$ is Martin-L\"of random with respect to $\mu$ if and only if there exists a positive constant $c$ such that for all $n$
\begin{align}
 K(x \upharpoonright n) > - \log \mu([x \upharpoonright n]) -c.
\end{align}
\end{thm}

\begin{proof}
See Theorem 6.2.3 in \cite{dh} for the case of the uniform measure. The extension to an arbitrary computable measure is straightforward.
\end{proof}

There is another characterization by unpredictability. If a sequence is random, the knowledge of the first $n$ bits of the sequence provides no useful information on the $(n+1)$-th bit. Therefore, if we bet on the future bits by utilizing the knowledge of the previous bits of the random sequence, there is no betting strategy by which we are able to make much money. A betting strategy is represented mathematically by a martingale, which is a crucial concept in the theory of stochastic process \cite{doob}. See \cite{nies,dh} for more details in the context of the theory of algorithmic randomness.

\section{Kac infinite chain model}
\label{sec:3}

We aim to understand how deterministic and irreversible macroscopic laws emerge from deterministic and reversible microscopic dynamics. The Kac ring model has often been used as an instructive model to demonstrate the macroscopic law as the law of large numbers \cite{kac,go,mns}. The model is also suitable as an example of an application of the randomness notion because the dynamical system is defined on infinite binary sequences.

\subsection{Model}
Let us consider the one-dimensional lattice $\mathbb{Z}$. For each site $i \in \mathbb{Z}$, there is one particle having a spin variable $\eta(i) \in \{ -1, 1\}$ and at most one scatterer. The occupation number of the scatterer at site $i$ is denoted by $y(i) \in \{ 0, 1\}$. For convenience, we set $x(i) = (1 + \eta(i))/2 \in \{ 0, 1\}$ for all $i \in \mathbb{Z}$ and think of them as dynamical variables. Then, a microscopic state of our model is represented by $(x,y) = (x(i), y(i))_{i \in \mathbb{Z}}$ and the state space is $\{ 0, 1\}^{\mathbb{Z}} \times \{ 0, 1\}^{\mathbb{Z}}$. A discrete-time deterministic dynamical system on $\{ 0, 1\}^{\mathbb{Z}} \times \{ 0, 1\}^{\mathbb{Z}}$ is defined by a function $\varphi: \{ 0, 1\}^{\mathbb{Z}} \times \{ 0, 1\}^{\mathbb{Z}} \to \{ 0, 1\}^{\mathbb{Z}} \times \{ 0, 1\}^{\mathbb{Z}}$ with
\begin{align}
\label{eq:dyn}
 \varphi(x,y)(i) = (x(i-1) + y(i-1) - 2 x(i-1) y(i-1), y(i)).
\end{align}
By using (\ref{eq:dyn}), we obtain the time evolution of the spin variables, $\eta(x,y)(i) = 2x(i)-1$, as
\begin{align}
\label{eq:spin}
 (\eta \circ \varphi^t)(x,y)(i) = [1-2y(i-1)] \dots [1-2y(i-t)] (2x(i-t)-1).
\end{align}
Thus, this dynamical system has the following interpretation. Let us prepare an initial configuration of spins and scatterers. For each time step, the configuration of scatterers remains unchanged and the particle at site $i$ jumps to the neighbor site $i+1$. Then, the spin $\eta(i)$ of the particle is flipped if the scatterer at site $i$ is present, $y(i)=1$, or it keeps its value if absent, $y(i)=0$.

The dynamical system is deterministic and invertible. In fact, the map
\begin{align}
 \varphi^{-1}(x,y)(i) = (x(i+1) + y(i) - 2 x(i+1) y(i), y(i))
\end{align}
is the inverse of $\varphi$. Obviously, this dynamics corresponds to jumps of particles to the left site. We discuss the details of the microscopic reversibility in section \ref{sec:4}. 

The dynamical system on $\{ 0, 1\}^{\mathbb{Z}} \times \{ 0, 1\}^{\mathbb{Z}}$ can be regarded as that on $\{ 0, 1 \}^{\mathbb{N}} \times \{ 0, 1 \}^{\mathbb{N}}$ and $\{ 0, 1 \}^{\mathbb{N}}$ by the encoding function $\iota$ and $\kappa$. Hereafter, $(x,y)$ (resp. $\varphi$) represents an element of $\{ 0, 1\}^{\mathbb{Z}} \times \{ 0, 1\}^{\mathbb{Z}}$, $\{ 0, 1 \}^{\mathbb{N}} \times \{ 0, 1 \}^{\mathbb{N}}$, or $\{ 0, 1 \}^{\mathbb{N}}$ (resp. the function on $\{ 0, 1\}^{\mathbb{Z}} \times \{ 0, 1\}^{\mathbb{Z}}$, $\{ 0, 1 \}^{\mathbb{N}} \times \{ 0, 1 \}^{\mathbb{N}}$, or $\{ 0, 1 \}^{\mathbb{N}}$) interchangeably. It is easy to show that if we think $\varphi$ as a function from $\{ 0, 1 \}^{\mathbb{N}}$ to $\{ 0, 1 \}^{\mathbb{N}}$, $\varphi$ is a total computable function on $\{ 0, 1 \}^{\mathbb{N}}$.

\begin{rem}
Our model is a variant of the Kac ring model \cite{kac}. The original model is defined on the ring of size $N$. We use the infinite chain model in this paper because the randomness notion in section \ref{sec:2} is sharply defined for infinite sequences. Therefore, Zermelo's recurrence paradox, which is a characteristic of finite systems, does not occur.

The model can be thought to be a dynamical system that consists of spin degrees of freedom with quenched scatterers because the configuration of scatterers does not change in time. In this paper, we include the scatterers in state variables for simplicity. See Remark \ref{rem:vl}
\end{rem}

\subsection{Measure-theoretic approach}

Let us imagine the situation we observe the system macroscopically. We introduce the following two macroscopic variables over $2N+1$ sites for a microscopic state $(x,y)$:
\begin{align}
\label{eq:macro}
 m_0^N (x,y)  = \frac{1}{2N+1} \sum_{i=-N}^N (2x(i)-1), \ \ m_1^N(x,y) = \frac{1}{2N+1} \sum_{i=-N}^N y(i).
\end{align}
If we observe the time evolution of the macroscopic variables $m(t) = (m_0(t), m_1(t))$, the variables obey a macroscopic law and relax to the equilibrium values. In fact, at each time step $t$, we assume that the up or down spins are scattered at a rate $m_1(t)$ for sufficiently large $N$. Then, the fraction of the up or down spins changes from $(1 \pm m_0(t))/2$ to $[1 \pm (1-2m_1(t))m_0(t)]/2$. Therefore, the average magnetization changes from $m_0(t)$ to $(1-2m_1(t))m_0(t)$. Because the average density of the scatterers is constant, the macroscopic law has the form $m(t) = \Phi^t(m(0))$ with $\Phi(m) = ([1-2m_1]m_0, m_1)$. 

This ``molecular chaos'' argument provides the form of the macroscopic law that the system should obey \textit{on average}. However, the hydrodynamic equations for fluids predict the macroscopic behavior of a single experiment, not just the ensemble average. The same holds true for this model. Suppose that initial microscopic states are sampled according to an initial probability measure corresponding to a nonequilibrium state. Then, the macroscopic law is understood as typical behavior with respect to the initial probability measure. This scenario is represented mathematically by the law of large numbers.

In statistical mechanics, if we have information on only the values of relevant macroscopic variables at the initial time, then one natural choice of an initial probability measure is the Gibbs measure corresponding to the initial macroscopic state \cite{zmr1,zmr2}. In the case of the Kac infinite chain model, the relevant macroscopic variables are the average magnetization $m_0$ and the average density of scatterers $m_1$. Then, the Gibbs measure in this case is the product of the Bernoulli measures $\mu_{(1+m_0)/2} \times \mu_{m_1}$ on $\{ 0, 1 \}^{\mathbb{Z}} \times \{ 0, 1 \}^{\mathbb{Z}}$, where $m=(m_0, m_1) \in [-1,1] \times [0,1]$ is an initial nonequilibrium state.

Under the above settings, the weak and strong laws of large numbers hold. Although the facts are widely known, we give complete proofs of the theorems in the following. Henceforth, we write $\mu_m = \mu_{(1+m_0)/2} \times \mu_{m_1}$ for notational simplicity and $\mathbb{E}[X]$ denotes the expectation value of a random variable $X$ with respect to $\mu_{m}$. For instance,
\begin{align}
\label{eq:bep}
 \mathbb{E}[ 2 x(i) - 1 ] = m_0, \ \ \mathbb{E}[ y(i) ] =m_1 \ \ \text{for} \ i \in \mathbb{Z}.
\end{align}

\begin{thm}[weak law of large numbers \cite{kac,go,mns}] \label{thm:wlln}
For any $T \in \mathbb{N}$ and any $\delta >0$,
\begin{align}
 \lim_{N \to \infty} \mu_m \left( \bigcup_{t=0}^T \bigcup_{i \in \{0, 1\}} \left\{ (x,y) \in \{ 0, 1\}^{\mathbb{Z}} \times \{ 0, 1\}^{\mathbb{Z}} : | (m_i^N \circ \varphi^t)(x,y) - \Phi^t_i(m)| > \delta \right\} \right) = 0.
\end{align}
\end{thm}

\begin{proof}

Fix $T \in \mathbb{N}$ and $\delta >0$. By the subadditivity of measures, it is enough to show that for any $t \in \{0, \dots, T\}$ and any $i \in \{0, 1\}$,
\begin{align}
\label{eq:enough}
 \lim_{N \to \infty} \mu_m \left(  | (m_i^N \circ \varphi^t)(x,y) - \Phi^t_i(m)| > \delta  \right) = 0.
\end{align}
First, we show that
\begin{align}
 \mathbb{E}[ (m_i^N \circ \varphi^t)(x,y) ] = \Phi_i^t(m).
\end{align}
By using (\ref{eq:spin}), (\ref{eq:macro}), (\ref{eq:bep}) and statistical independence of $x(i)$ and $y(j)$, we have
\begin{align}
 \mathbb{E}[ (m_0^N \circ \varphi^t)(x,y) ] &= \frac{1}{2N+1} \sum_{i=-N}^{N} \mathbb{E}[(1-2y(i-1)) \dots (1-2y(i-t))(2x(i-t)-1)]
 \notag \\
 &= \frac{1}{2N+1} \sum_{i=-N}^{N} \mathbb{E}[1-2y(i-1)] \dots \mathbb{E}[1-2y(i-t)] \mathbb{E}[2x(i-t)-1]
 \notag \\
 &= (1-2m_1)^t m_0 = \Phi_0^t(m).
\end{align}
$\mathbb{E}[(m_1^N \circ \varphi^t)(x,y)] = m_1 = \Phi^t_1(m)$ is obvious.
Next, we evaluate the second moments of $(m_i^N \circ \varphi^t)(x,y)$.
\begin{align}
 \mathbb{E}\left[\left((m_0^N \circ \varphi^t)(x,y)\right)^2\right] &= \frac{1}{(2N+1)^2} \sum_{i,j=-N}^N \mathbb{E}[ (1-2y(i-1)) \dots (1-2y(i-t)) 
 \notag \\
 & \ \ \ \  \times  (1-2y(j-1)) \dots (1-2y(j-t)) ] \cdot \mathbb{E}[ (2x(i-t)-1)(2x(j-t)-1)]
 \notag \\
 &= \frac{1}{(2N+1)^2} \sum_{k=-2N}^{2N} (2N+1 -|k|)  \ \mathbb{E}[ (1-2y(0)) \dots (1-2y(t-1)) 
 \notag \\
 & \ \ \ \  \times  (1-2y(k)) \dots (1-2y(k+t-1)) ] \cdot \mathbb{E}[ (2x(0)-1)(2x(k)-1)]
 \notag \\
 &= \frac{1}{(2N+1)^2} \left[ (2N+1) + 2 m_0^2 \ \sum_{k=1}^{2N} (2N+1 -k) (1-2m_1)^{2 \min \{t,k \}} \right].
\end{align}
We have used the translation invariance of $\mu_{m}$ and statistical independence of random variables at different sites. We take $N$ such that $T \leq 2N$. Then, we have
\begin{align}
 \sum_{k=1}^{2N} (2N+1 -k) (1-2m_1)^{2 \min \{t,k \}} &= \sum_{k=1}^{t} (2N+1 -k) (1-2m_1)^{2k}  + \sum_{k=t+1}^{2N} (2N+1 -k) (1-2m_1)^{2t}
 \notag \\
 &\leq (2N+1)t + \frac{1}{2} (1-2m_1)^{2t} (2N+1)^2.
\end{align}
Therefore,
\begin{align}
 \mathrm{Var}[(m_0^N \circ \varphi^t)(x,y)] \leq \frac{1 + 2 m_0^2 t}{2N+1}
\end{align}
for $t \in \{0, \dots, T \}$ and $T \leq 2N$. Additionally, we obtain the variance of $m_1^N$,
\begin{align}
 \mathrm{Var}[(m_1^N \circ \varphi^t)(x,y)] = \frac{1-m_1^2}{2N+1}.
\end{align}
By using Chebyshev's inequality, we have
\begin{align}
\label{eq:cinq}
 \mu_{m} \left( | (m_i^N \circ \varphi^t)(x,y) - \Phi^t_i(m)| > \delta  \right) \leq \frac{\mathrm{Var}[(m_i^N \circ \varphi^t)(x,y)]}{\delta^2} \leq \frac{C}{\delta^2(2N+1)}
\end{align}
with a constant $C$ independent of $N$, which implies (\ref{eq:enough}). \qed
\end{proof}

We have the strong form of the law of large numbers from the inequality (\ref{eq:cinq}) and the first Borel-Cantelli lemma (Theorem \ref{thm:bc}).

\begin{thm}[strong law of large numbers]
\label{thm:slln}
For any natural number $T \in \mathbb{N}$,
\begin{align}
 \mu_m \left( \lim_{N \to \infty} (m_i^N \circ \varphi^t)(x,y) = \Phi_i^t(m) \ \text{for all $i \in \{0, 1\}$ and $t \in \{0, \dots, T\}$} \right) = 1.
\end{align}
\end{thm}

\begin{proof}
For $k \in \mathbb{N}_{>0}$, we set
\begin{align}
 C_{N,k} = \bigcup_{t=0}^T \bigcup_{i \in \{0, 1\}} & \left\{ (x,y) \in \{ 0, 1\}^{\mathbb{Z}} \times \{ 0, 1\}^{\mathbb{Z}} :  |(m_i^{N} \circ \varphi^t)(x,y) - \Phi^t_i(m)| > \frac{1}{k} \right\}.
\end{align}
From (\ref{eq:cinq}), we have
\begin{align}
 \mu_{m} (C_{N^2,k}) \leq \frac{2(T+1)Ck^2}{2N^2+1}
\end{align}
for $N$ satisfying $T \leq 2N$. Therefore,
\begin{align}
\label{eq:bct}
 \sum_{N=0}^{\infty} \mu_{m} (C_{N^2,k}) \leq \sum_{N=0}^{\lceil T/2 \rceil} \mu_{m} (C_{N^2,k}) + \sum_{N= \lceil T/2 \rceil}^{\infty} \frac{2(T+1)Ck^2}{(2N^2+1)} < \infty.
\end{align}
By the first Borel-Cantelli lemma (Theorem \ref{thm:bc}) and Proposition \ref{prop:null}, we have
\begin{align}
 \mu_m \left( \bigcup_{k \in \mathbb{N}_{>0}} \bigcap_{N \in \mathbb{N}} \bigcup_{M \geq N} C_{M^2, k} \right) = 0.
\end{align}
This means that for $\mu_m$-almost all configuration $(x,y)$, the subsequence $ ((m_i^{N^2} \circ \varphi^t)(x,y))_{N \in \mathbb{N}}$ converges to $\Phi^t_i(m)$ for $t \in \{0, \dots, T \}$ and $i \in \{0,1\}$

Next, we show the convergence of the whole sequence. For any natural numbers $L,M$ and a real number $p$ with $p \leq (2L+1)/(2M+1) \leq 1$, we have
\begin{align}
\label{eq:ineq}
 (m_0^{L} \circ \varphi^t)(x,y) - (m_0^{M} \circ \varphi^t)(x,y) &= \left( 1 - \frac{2L+1}{2M+1} \right) (m_0^{L} \circ \varphi^t)(x,y)
 \notag \\
 & \ \  - \frac{1}{2M+1} \left( \sum_{i=-M}^{-L-1} (2x(i)-1) + \sum_{i=L+1}^{M} (2x(i)-1) \right)
 \notag \\
 &\leq 1-p + \frac{2(M-L)}{2M+1} \leq 2(1-p),
\end{align}
where we have used $m_0^L \in [-1,1]$ and $2x(i)-1 \in \{ -1, 1\}$. We consider a natural number $K$ such that $N^2 \leq K \leq (N+1)^2$. If we take $L=N^2$, $M=K$ and $p=p_N=(2N^2+1)/(2(N+1)^2+1)$ first and take $L=K$, $M=(N+1)^2$ and $p=p_N$ second, the inequality (\ref{eq:ineq}) gives
\begin{align}
 (m_0^{N^2} \circ \varphi^t)(x,y) - 2( 1 - p_N ) \leq (m_0^{K} \circ \varphi^t)(x,y) \leq (m_0^{(N+1)^2} \circ \varphi^t)(x,y) + 2( 1 - p_N ).
\end{align}
These inequalities also hold for $(m_1^{K} \circ \varphi^t)(x,y)$. Since $p_N \to 1$ as $N \to \infty$, for $\mu_m$-almost all $(x,y)$, the whole sequence $((m_i^{N} \circ \varphi^t)(x,y))_{N \in \mathbb{N}}$ converges to $\Phi_i^t(m)$ for $i \in \{ 0, 1\}$ and $t \in \{ 0, \dots, T\}$. \qed
\end{proof}

\begin{rem}
A similar analysis leads to the weak and strong law of large numbers for the microscopic dynamics $\varphi^{-1}$ and positive integers $T>0$. This is a consequence of the microscopic reversibility and the statistical property of the initial measure. If an initial configuration has no correlation between sites, whether the microscopic dynamics is $\varphi$ or $\varphi^{-1}$, which corresponds to the direction of movement of the spins, is irrelevant to the validity of the macroscopic relaxation.
\end{rem}

\subsection{Algorithmic randomness approach}

We reformulate the law of large numbers associated with the macroscopic law as properties of individual microscopic states. The concept of algorithmic randomness introduced in section \ref{sec:2} helps us to do that. By using the randomness notion, we have the following theorem.

\begin{thm}[effective strong law of large numbers]
\label{thm:elln}
Let $m=(m_0, m_1)$ be computable reals and $T \in \mathbb{N}$. If $(x,y) \in \{ 0, 1 \}^{\mathbb{Z}} \times \{ 0, 1 \}^{\mathbb{Z}}$ is Martin-L{\"o}f random with respect to $\mu_{m}$,
\begin{align}
 \lim_{N \to \infty} (m_i^N \circ \varphi^t)(x,y) = \Phi_i^t(m)
\end{align}
for all $i \in \{0, 1\}$ and $t \in \{0, \dots, T\}$.
\end{thm}

\begin{proof}
Fix $T \in \mathbb{N}$. Since $\Phi_i^t(m)$ are computable reals, $(C_{N^2,k})_{N \in \mathbb{N}}$ is c.e. open uniformly in $N$. By (\ref{eq:bct}), it is a Solovay test. Therefore, if $(x,y) \in \{ 0, 1 \}^{\mathbb{Z}} \times \{ 0, 1 \}^{\mathbb{Z}}$ is ML random with respect to $\mu_{m}$, then $(x,y) \in \bigcup_{N \in \mathbb{N}} \bigcap_{M \geq N} (C_{N^2,k})^c$. Because $k$ is arbitrary, the subsequence $((m_i^N \circ \varphi^t)(x,y))_{N \in \mathbb{N}}$ converges to $\Phi_i^t(m)$ for any $i \in \{0, 1\}$ and $t \in \{ 0, \dots, T \}$. The proof of the convergence of the whole sequence is the same as the proof of Theorem \ref{thm:slln}. \qed
\end{proof}

According to Theorem \ref{thm:elln}, the algorithmic randomness of a microscopic state is a sufficient condition that the microstate obeys the macroscopic relaxation law. Since the set of all ML random microstates has measure one (see Theorem \ref{thm:full}), this sufficient condition is not too strong from a viewpoint of measure-theoretic typicality. In particular, the strong law of large numbers (Theorem \ref{thm:slln}) follows from the effective law. We stress that the effective law of large numbers holds for a wide class of models. We discuss the generality of our result in \ref{subsec:rt}.

\begin{rem} \label{rem:vl}
Van Lambalgen's theorem \cite{vl} implies that $(x,y)$ is ML $\mu_m$-random if and only if $y$ is ML $\mu_{m_1}$-random and $x$ is ML $\mu_{(1+m_0)/2}$-random with \textit{oracle} $y$. Therefore, Theorem \ref{thm:elln} insists that for a given ML $\mu_{m_1}$-random configuration of quenched scatterers $y$, $\mu_{(1+m_0)/2}$-random microstates with oracle $y$ satisfy the macroscopic law. We note that for a $\mu_{(1+m_0)/2}$-random element $x$ and $\mu_{m_1}$-random element $y$, the pair $(x,y)$ does not necessarily obey the macroscopic law. For instance, if $(1+m_0)/2 = m_1$ and $x$ is $\mu_{(1+m_0)/2}$-random, $(x,x)$ violates the law.
\end{rem}

\section{Entropy and the zeroth law of thermodynamics}
\label{sec:4}

Entropy is a fundamental concept in various fields such as thermodynamics, statistical physics, information theory and dynamical systems theory. Each type of entropy has a different role. We investigate the Boltzmann entropy quantifying irreversibility on transitions between macroscopic states. 

\subsection{Shannon entropy}

Before considering the Boltzmann entropy, we review basic properties of the Shannon entropy for convenience, which is an information-theoretic quantity characterizing the optimal compression rate in the information source coding problem \cite{ct}. 

\begin{rem}
In this subsection, the configurations $(x,y)$, the probability measures and the microscopic dynamics $\varphi$ are regarded as ones defined on $\{ 0, 1\}^{\mathbb{N}} \times \{ 0, 1\}^{\mathbb{N}}$, not on $\{ 0, 1\}^{\mathbb{Z}} \times \{ 0, 1\}^{\mathbb{Z}}$, by the encoding function $\iota$ in section \ref{subsec:mlr}.
\end{rem}

\begin{defin}[Shannon entropy rate, self-entropy rate]
The \textit{Shannon entropy rate} of the joint probability measure $\mu$ on $\{ 0, 1 \}^{\mathbb{N}} \times \{ 0, 1 \}^{\mathbb{N}}$ is defined as
\begin{align}
 \bar{H}(\mu) = \limsup_{n \to \infty} - \frac{1}{n} \sum_{(\sigma, \tau) \in 2^n \times 2^n} \mu ([\sigma] \times [ \tau ]) \ln \mu([\sigma] \times [ \tau ]).
\end{align}
The self-entropy rate of $(x,y) \in \{ 0, 1 \}^{\mathbb{N}} \times \{ 0, 1 \}^{\mathbb{N}}$ with respect to $\mu$ is defined as
\begin{align}
 \bar{H}_{\mu}(x,y) = \limsup_{n \to \infty} - \frac{1}{n} \ln \mu( [ x \upharpoonright n] \times [ y \upharpoonright n]).
\end{align}
\end{defin}

A straightforward calculation provides 
\begin{align}
\bar{H}(\mu_{m}) =  h \left( \frac{1+m_0}{2} \right) + h(m_1),
\end{align}
where $h(p) = - p \ln p - (1-p) \ln (1-p)$ ($p \in [0,1]$) is the binary entropy function. For $x \in \{ 0, 1 \}^{\mathbb{N}}$ and $n \in \mathbb{N}$, set $N(x,n)= | \{ i : x(i) = 1, 0 \leq i \leq n-1 \} |$. If $(x,y) \in \{ 0, 1 \}^{\mathbb{N}} \times \{ 0, 1 \}^{\mathbb{N}}$ satisfies
\begin{align}
 \lim_{n \to \infty} \frac{N(x,n)}{n} = p_x, \ \ \lim_{n \to \infty} \frac{N(y,n)}{n} = p_y,
\end{align}
the self-entropy rate of $(x,y)$ with respect to $\mu_{m}$ is given by
\begin{align}
 \bar{H}_{\mu_{m}}(x,y) =&  - p_x \ln \left( \frac{1+m_0}{2} \right) - (1-p_x) \ln \left( \frac{1- m_0}{2} \right)  - p_y \ln m_1 - (1-p_y) \ln (1-m_1).
\end{align}
In particular, for any random element $(x,y) \in \{ 0, 1 \}^{\mathbb{N}} \times \{ 0, 1 \}^{\mathbb{N}}$ with respect to $\mu_{m}$, 
\begin{align}
 \bar{H}_{\mu_{m}}(x,y)  = \bar{H}(\mu_{m}),
\end{align}
because $p_x = (1 + m_0)/2$ and $p_y = m_1$. This type of statement is referred to as the effective version of the \textit{asymptotic equipartition property}.

For deterministic and reversible dynamical systems, the Shannon entropy of the probability measure describing the system does not provide useful information on irreversibility. If we define the probability measure at time $t$ starting from the initial measure $\mu_m$ as $\mu_{m,t} = \mu_m \varphi^{-t}$, the Shannon entropy rate is invariant under the time evolution, that is, $\bar{H}(\mu_{m}) = \bar{H}(\mu_{m,t})$. This invariance remains true for random elements. In fact, since the initial segment of the first and second components of $\varphi^t(x,y) \in \{ 0, 1 \}^{\mathbb{N}} \times \{ 0, 1 \}^{\mathbb{N}}$, $x_t \upharpoonright n$ and $y_t \upharpoonright n$, depend only on $x \upharpoonright n+2t$ and $y \upharpoonright n+2t$ (the factor 2 comes from the way of encoding $\iota$ from $\{ 0, 1\}^{\mathbb{Z}}$ to $\{ 0, 1 \}^{\mathbb{N}}$), the inclusion relation
\begin{align}
\label{eq:inc}
 [ x \upharpoonright n+2t] \times [ y \upharpoonright n+2t] \subseteq \varphi^{-t} ([ x_t \upharpoonright n] \times [ y_t \upharpoonright n] ) \subseteq [ x \upharpoonright n-2t] \times [ y \upharpoonright n-2t]
\end{align}
holds. Then,
\begin{align}
 - \frac{1}{n} \ln \mu_{m}([ x \upharpoonright n+2t] \times [ y \upharpoonright n+2t]) &\leq - \frac{1}{n} \ln \mu_{m,t} ([ x_t \upharpoonright n] \times [ y_t \upharpoonright n] )
 \notag \\
 & \leq - \frac{1}{n} \ln \mu_{m}([ x \upharpoonright n-2t] \times [ y \upharpoonright n-2t]).
\end{align}
For any random element $(x,y) \in \{ 0, 1 \}^{\mathbb{N}} \times \{ 0, 1 \}^{\mathbb{N}}$, the terms on the left- and right-hand sides converge to $\bar{H}(\mu_{m})$. Therefore, the self-entropy rate of $\varphi^t(x,y)$ with respect to $\mu_{m,t}$ exists and equals that of $(x,y)$ with respect to $\mu_{m}$:
\begin{align}
 \bar{H}_{\mu_{m,t}}(\varphi^t(x,y)) = \lim_{n \to \infty} - \frac{1}{n} \ln \mu_{m,t} ([ x_t \upharpoonright n] \times [ y_t \upharpoonright n]) = \bar{H}(\mu_{m}) = \bar{H}(\mu_{m,t}).
\end{align}

\subsection{Boltzmann entropy and the zeroth law of thermodynamics}

Because the Shannon entropy does not change in time in reversible dynamical systems, we need another quantity to characterize the macroscopic irreversibility. According to Boltzmann's idea, the asymmetry of the direction of time in the macroscopic behavior emerges from the large differences between the number of microstates consistent with macrostates. Since the number of microstates corresponding to a macrostate is proportional to the probability of the macrostate under the uniform measure $\lambda$ and different macrostates usually have exponentially different probabilities, it is reasonable to introduce the rate function in the large deviation theory as the Boltzmann entropy.

\begin{defin}[Boltzmann entropy]
The \textit{Boltzmann entropy} of a macroscopic state $m=(m_0, m_1)$ is defined as
\begin{align}
 S_{B}(m) = \lim_{\delta \downarrow 0} \lim_{N \uparrow \infty} \frac{1}{2N+1} \ln \lambda \times \lambda \left( \{ (x,y) \in \{ 0, 1 \}^{\mathbb{Z}} \times \{ 0, 1 \}^{\mathbb{Z}} : |m_0^N(x) - m_0| \leq \delta,  |m_1^N(x) - m_1| \leq \delta \} \right),
\end{align}
where $\lambda$ is the uniform measure on $\{ 0, 1 \}^{\mathbb{Z}}$.
\end{defin}

The following scenario is well-known \cite{lebowitz}: An initial microstate in a nonequilibrium macrostate with low Boltzmann entropy evolves \textit{typically} toward macrostates with higher entropy and finally reaches the equilibrium state with the maximum entropy. 

Although at first sight, it explains the macroscopic irreversibility qualitatively, it should be noted that we must suppose an initial probability measure in order to argue the \textit{typical} macroscopic behavior. The above scenario is certainly true if we assume that initial microstates are chosen according to the microcanonical measure or the Gibbs measure. In fact, by Stirling's formula, we have
\begin{align}
 S_B(m) = - 2 \ln 2 + \bar{H}(\mu_{m}).
\end{align}
If we prepare initial microstates according to the Gibbs measure $\mu_m$, then the initial macrostate is $m$ and the macrostate evolves according to the law $\Phi$ with probability one according to Theorem \ref{thm:slln}. Then, the Boltzmann entropy difference is typically positive: 
\begin{align}
S_B(\Phi^t(m)) - S_B(m) = \bar{H}(\mu_{\Phi^t(m)}) - \bar{H}(\mu_m) > 0 \ \text{for} \ t > 0.
\end{align}
We can reformulate the argument from a viewpoint of randomness. If an initial state $(x,y) \in \{ 0, 1 \}^{\mathbb{Z}} \times \{ 0, 1 \}^{\mathbb{Z}}$ is ML $\mu_{m}$-random, the macrostate at the initial time is $m$ and at time $t$ is $\Phi^t(m)$ (see Theorem \ref{thm:elln}). Then, the Boltzmann entropy increases over time. This is the zeroth law of thermodynamics for algorithmic random microstates.

\section{Microscopic reversibility and anti-Boltzmann behavior}
\label{sec:5}

\subsection{Microscopic reversibility}

The microscopic dynamics $\varphi : \{ 0, 1 \}^{\mathbb{Z}} \times \{ 0, 1 \}^{\mathbb{Z}} \to \{ 0, 1 \}^{\mathbb{Z}} \times \{ 0, 1 \}^{\mathbb{Z}}$ is invertible. This property is referred to as \textit{microscopic reversibility}. Let us define a time-reversal transformation $\pi : \{ 0, 1 \}^{\mathbb{Z}} \times \{ 0, 1 \}^{\mathbb{Z}} \to \{ 0, 1 \}^{\mathbb{Z}} \times \{ 0, 1 \}^{\mathbb{Z}}$ by
\begin{align}
 (\pi(x,y))(i) = (x(-i), y(-i-1)).
\end{align}
The time-reversal transformation is an involution $\pi^2 = 1$ and is totally computable. The microscopic reversibility is represented by $\pi \circ \varphi = \varphi^{-1} \circ \pi$. We note that $m_0^N(\pi(x,y)) = m_0^N(x,y)$ and $m_1^N(\pi(x,y))=m_1^N(x,\Sigma^{-1}(y))$. Here, $\Sigma(x)(i) = x(i+1)$ is the shift map on $\{ 0, 1\}^{\mathbb{Z}}$. In particular, the time-reversal transformation does not affect the macrostate.

\subsection{Irreversible information loss}

If there is a microscopic trajectory $( \varphi^t(x,y) )_{t \in \{0, 1, \dots, T\}}$ whose macroscopic trajectory is $( \Phi^t(m) )_{t \in \{0, 1, \dots, T\}}$, then the time-reversed one $((\varphi^t \circ \pi \circ \varphi^T)(x,y))_{t \in \{0, 1, \dots, T\}} = ( (\pi \circ \varphi^{T-t})(x,y))_{t \in \{0, 1, \dots, T\}}$ is macroscopically observed as $(\Phi^{T-t}(m))_{t \in \{0, 1, \dots, T\}}$. Loschmidt inquired how the above consequence of the microscopic reversibility is consistent with the macroscopic irreversibility. This question is called the reversibility paradox problem.

Sasa and Komatsu introduced the irreversible information loss quantifying the asymmetry between the trajectory $( \varphi^t(x,y) )_{t \in \{0, 1, \dots, T\}}$ and the time-reversed one $( (\pi \circ \varphi^{T-t})(x,y))_{t \in \{0, 1, \dots, T\}}$, and investigated the relation to the Boltzmann entropy change \cite{sk}. Following this idea, we define the rate of \textit{irreversible information loss} as
\begin{align}
 I_{\mu_{m}, t}(x,y) = \limsup_{n \to \infty} - \frac{1}{n} \ln \frac{\mu_{m,t} ( \pi ([ x_t \upharpoonright n] \times [y_t \upharpoonright n ]) )}{ \mu_{m,t} ( [ x_t \upharpoonright n  ] \times [ y_t \upharpoonright n] ) },
\end{align}
where the dynamical system is regarded as one on $\{0,1\}^{\mathbb{N}} \times \{ 0, 1\}^{\mathbb{N}}$. The positivity of the irreversible information loss of a microstate $(x,y)$ implies the exponential difference between the probabilities at time $t$ of the microstate $\varphi^t(x,y)$ and the time-reversed one $(\pi \circ \varphi^t)(x,y)$. Then, it explains how difficult it is to prepare the time-reversed state $(\pi \circ \varphi^t)(x,y)$ relative to the state $\varphi^t(x,y)$ in the measure-theoretic sense. We note that this argument is different from the standard one on the reversibility paradox indicating the practical impossibility of the time-reverse transformation.

We can easily calculate the above quantity for random states $(x,y)$ as follows. The microscopic reversibility implies that $\varphi^{-t} \circ \pi \circ \varphi^t = \pi \circ \varphi^{2t}$. If $(x,y)$ is ML $\mu_{m}$-random, by Theorem \ref{thm:elln}, we have
\begin{align}
 \lim_{n \to \infty} \frac{N(x_{2t},n)}{n} = \frac{1 + \Phi_0^{2t}(m)}{2}, \ \ \lim_{n \to \infty} \frac{N(y_{2t},n)}{n} = \Phi_1^{2t}(m).
\end{align}
The same type of inclusion relation as (\ref{eq:inc}) implies
\begin{align}
 \lim_{n \to \infty} - \frac{1}{n} \ln \mu_{m,t} ( \pi ([ x_t \upharpoonright n] \times [y_t \upharpoonright n ]) ) =&  - \frac{1 + \Phi_0^{2t}(m)}{2} \ln \left( \frac{1+m_0}{2} \right) - \frac{1 - \Phi_0^{2t}(m)}{2} \ln \left( \frac{1- m_0}{2} \right) 
 \notag \\
  & \ \  - \Phi_1^{2t}(m) \ln m_1 - (1-\Phi_1^{2t}(m)) \ln (1-m_1) .
\end{align}
Therefore, the rate of the irreversible information loss of a random element $(x,y)$ is given by
\begin{align}
 I_{\mu_{m},t}(x,y) = \frac{m_0 - \Phi_0^{2t}(m)}{4} \ln \frac{1+m_0}{1-m_0} = \frac{1 - (1-2m_1)^{2t}}{4} m_0\ln \frac{1+m_0}{1-m_0}.
\end{align}
If $m_0 \neq 0$, $m_1 \neq 1/2$, and $t >0$, then $I_{\mu_{m},t}(x,y) > 0$. This result implies the measure-theoretic difficulty of preparing the time-reversed state relative to the random state. Moreover, by explicit calculation, we have
\begin{align}
 I_{\mu_{m},t}(x,y) - (S_B(\Phi^t(m)) - S_B(m)) \geq 0.
\end{align}
The above equality holds for $m_0=0$, $m_1 \in \{ 0, 1/2 \}$, or $t=0$. That is, the degree of difficulty is greater than the Boltzmann entropy change in this case.

\subsection{Violation of the macroscopic law and nonrandomness of time-reversed states}
\label{subsec:53}

Suppose that an initial probability measure is $\mu_m$ and $(x,y) \in \{ 0, 1\}^{\mathbb{Z}} \times \{ 0, 1\}^{\mathbb{Z}}$ is ML $\mu_{m}$-random. By the conservation of ML randomness (see Theorem \ref{thm:cons}), $\varphi^t(x,y)$ is ML random with respect to the probability measure at time $t$, $\mu_{m,t}$. That is, the randomness of the initial microstate is preserved under the dynamics. An intriguing question is whether the time-reversed state $(\pi \circ \varphi^t)(x,y)$ is ML $\mu_{m,t}$-random or not. The microscopic reversibility implies that the macroscopic evolution starting from the state $(\pi \circ \varphi^t)(x,y)$ does not obey the macroscopic law $\Phi$. Therefore, $(\pi \circ \varphi^t)(x,y)$ is not $\mu_{m,t}$-random.

\begin{thm}[non-randomness of time-reversed state]
\label{thm:nrtr}
Let $m=(m_0, m_1)$ be computable reals and $m_0 \neq 0$, $m_1 \not\in \{ 0, 1/2 \}$. For any $\mu_{m}$-random element $(x,y) \in \{ 0, 1 \}^{\mathbb{Z}} \times \{0, 1\}^{\mathbb{Z}}$ and $t \in \mathbb{N} \backslash \{ 0 \}$, $(\pi \circ \varphi^t)(x,y)$ is not ML $\mu_{m,t}$-random.
\end{thm}

\begin{proof}
Assume $(\pi \circ \varphi^t)(x,y)$ is ML $\mu_{m,t}$-random for $t \in \mathbb{N} \backslash \{ 0 \}$. By the conservation of ML randomness (Theorem \ref{thm:cons}), $(u,v)=(\varphi^{-t} \circ \pi \circ \varphi^{t})(x,y) = (\pi \circ \varphi^{2t})(x,y)$ is ML $\mu_{m}$-random. From Theorem \ref{thm:elln}, for $s \in \{ 0, \dots, 2t \}$,
\begin{align}
 m_i^N(\varphi^s(u,v)) &= m_i^N((\varphi^s \circ \pi \circ \varphi^{2t})(x,y))
 \notag \\
 &= m_i^N((\pi \circ \varphi^{2t-s})(x,y))
 \notag \\
 &\to \Phi_i^{2t-s}(m) \ \ \text{as} \ N \to \infty.
\end{align}
Since $m_0 \neq 0$, $m_1 \not\in \{ 0, 1/2 \}$ and $t \neq 0$, there exists $s \in \{ 0, \dots, 2t \}$ such that $\Phi_0^{2t-s}(m) \neq \Phi_0^s(m)$ (take $s(\neq t)$). This is a contradiction. \qed
\end{proof}

This consideration leads to the following argument. $(u,v)=(\varphi^{-T} \circ \pi \circ \varphi^{T})(x,y) = (\pi \circ \varphi^{2T})(x,y)$ is ML $\mu_{m, 2T}\pi$-random if and only if $(x,y)=(\pi \circ \varphi^{2T})(u,v)$ is ML $\mu_{m}$-random. As in the proof of Theorem \ref{thm:nrtr},
\begin{align}
\label{eq:ab}
 \lim_{N \to \infty} m_i^N(\varphi^t(u,v)) = \Phi_i^{2T-t}(m)
\end{align}
for $t \in \{ 0, \dots, 2T \}$ and ML $\mu_{m, 2T} \pi$-random element $(u,v) \in \{ 0, 1 \}^{\mathbb{Z}} \times \{0,1\}^{\mathbb{Z}}$. Therefore, if we observe the macroscopic time evolution starting from a random microscopic state with respect to the \textit{initial} probability measure $\mu_{m, 2T} \pi$, the system exhibits the time-reversed behavior of the original macroscopic law $\Phi^t$. In particular, the Boltzmann entropy along the typical macroscopic trajectory \textit{decreases} monotonically: 
\begin{align}
 S_B(\Phi^{2T-t}(m)) - S_B( \Phi^{2T}(m)) < 0 \ \ \text{for} \ \ 0 < t \leq 2T.
\end{align}
Thus, typical macroscopic behavior depends on the choice of an initial probability measure. Even if an initial macroscopic state is given, the initial probability measure representing the state is not unique. Therefore, we have to demonstrate why we regard the Gibbs measure as important. See \ref{sssec:natural} for further discussion.

\section{Concluding remarks}
\label{sec:6}
\subsection{Open problems}

\subsubsection{What is a natural choice of measure?}
\label{sssec:natural}

We need to choose an initial probability measure to state a probabilistic law of large numbers for an irreversible macroscopic law. In this paper, we have chosen the Gibbs measure because it works well in many examples in statistical physics. Then, the system evolves typically so that the entropy increases monotonically to equilibrium. In contrast, as shown in \ref{subsec:53}, if we choose another initial measure carefully, the entropy decreases along the typical macroscopic evolution with respect to the measure. To elucidate the origin of macroscopic irreversibility, we have to clarify the difference between these two measures and to demonstrate why the measures under which the entropy increases are realized in our world. The problem also occurs when we discuss the effective law of large numbers since the notion of randomness formalizes typical states \textit{under a given probability measure}.

\subsubsection{What is the physical meaning of nonrandom states satisfying the macroscopic law?}

We have shown that random microstates satisfy the macroscopic law. However, the reverse is not generally true because the condition on the violation of the macroscopic law is just a part of Martin-L{\"o}f tests. That is to say, there are microscopic states satisfying the macroscopic law but not passing another ML test. Little is known about the physical meaning of such ML tests, and therefore also of nonrandom microstates satisfying the macroscopic law.

\subsubsection{Is algorithmic randomness really relevant to statistical physics?}
In the theory of algorithmic randomness, there are various classes of randomness according to the level of computability imposed on null sets in addition to the Martin-L\"of randomness (see Chapter 7 of \cite{dh} for example). In any case, we take into account \textit{all} effective null sets or corresponding statistical tests. However, all these tests are not necessarily realizable in physical experiments. Therefore, one may say that the theory of algorithmic randomness is unnecessary for the foundation of statistical physics.

A critical problem here is to identify the class of null sets associated with macroscopic properties. To consider the problem, let us recall the argument in \ref{subsec:mlr} motivating the definition of the Martin-L\"of randomness. We have seen that the law of large numbers does not provide a sufficient characterization of randomness. Even if we add another law such as the law of the iterated logarithm to the requirement of randomness, we may find other probabilistic laws having probability one and the requirement may turn out to be insufficient. Avoiding these difficulties, the theory of algorithmic randomness considers \textit{all} effective statistical laws and as a result clarifies a rich structure of randomness such as the equivalence between measure-theoretic typicalness, incompressibility and unpredictability. When we attempt to specify statistical laws involved with macroscopic properties, the above idea may be useful and there may be a deep connection between algorithmic randomness and statistical physics.

\subsection{Related topics}
\label{subsec:rt}
In this paper, we have shown that the algorithmic randomness of microscopic states is a sufficient condition of macroscopic relaxation in the Kac chain model. We expect to extend the theorem to a wider class of models. To prove the effective law of large numbers, we need the upper bound on the probability of the sequence of sets involved with the violation of the macroscopic law that tends to zero in the thermodynamic limit and the computable enumerability of the sequence. The former condition follows from a purely measure-theoretic argument. As long as we focus on macroscopic properties, the latter condition is also expected to be satisfied. For instance, there are deterministic and reversible dynamical systems with particle conservation that exhibit diffusive behavior in the sense of the law of large numbers \cite{raphael1,raphael2}. It is possible to extend our results to these models.

The models we refer above are cellular automata, that is, the dynamical systems on infinite lattices with local rules. The Martin-L\"of randomness in \ref{subsec:mlr} is defined on $\{ 0, 1 \}^{\mathbb{N}}$ and can be applied to only the dynamical systems on discrete state spaces. Recently, the notion of randomness has been generalized to computable metric spaces \cite{gacs2,hr} and applied to the dynamical system theory \cite{ghr10,ghr11}. Applying the theory to statistical physics is an important problem.

Another direction of future study is to generalize the notion of randomness to quantum systems. As in classical settings, the quantum Kolmogorov complexity of a quantum state is defined as the length of the shortest program outputting the description of the state \cite{gacs3,vitanyi}. The notion of Martin-L\"of random quantum state and the relation to the quantum Kolmogorov complexity has been investigated only recently \cite{ns17}. In either case, the algorithmic randomness theory of quantum systems has not yet been sufficiently studied compared to the case of classical systems. An example of the application of quantum randomness is the typicality of thermal equilibrium states \cite{tasaki,gltz,psw,sugita,reimann}. Although there are various mathematical formulations of the typicality of thermal equilibrium in quantum systems, they all state that almost all quantum pure states in a Hilbert space spanned by a set of the energy eigenstates represent thermal equilibrium. With the analogy to the argument in classical systems, we expect that \textit{random quantum states represent thermal equilibrium}. A more challenging theoretical issue in this context is the relation between the algorithmic randomness and the eigenstate thermalization hypothesis (ETH). The ETH insist that \textit{all} the energy eigenstates in an energy shell represent thermal equilibrium with the energy \cite{rdo}. The ETH is regarded as a plausible sufficient condition of thermalization in isolated quantum systems. The thermodynamic structure such as the fluctuation theorem and second law of thermodynamics that has been studied on the basis of the Gibbs state is being re-examined for the energy eigenstates \cite{iks,kis}. We anticipate that it is important to study the ETH from a viewpoint of algorithmic randomness.

Finally, we comment on the unpredictability aspect of algorithmic randomness associated with thermodynamics and statistical physics. The proofs and arguments in this paper are based on the measure-theoretic typicalness aspect of randomness. However, as explained in section \ref{subsec:rob}, the randomness notion has several characterizations such as incompressibility and unpredictability. We should further develop these aspects of statistical physics. For instance, the notion of martingale, which captures the unpredictability aspect of randomness, has not been well studied in statistical physics. Only recently, the martingale property of exponentiated entropy production in stochastic thermodynamics has been investigated \cite{cg,nrj}. A more challenging task is to investigate the fundamental assumption of statistical physic such as the principle of equal a priori probability from the viewpoint of the martingale property. For instance, the game-theoretic probability theory \cite{sv} provides a new formulation of limit theorems in probability theory such as the law of large numbers and the central limit theorem by utilizing only the betting game without the probabilistic structure. The idea in the game-theoretic probability theory that the probability emerges from the martingale property may be useful to study this problem.

\begin{acknowledgements}
The authors thank Naoto Shiraishi and Takahiro Sagawa for their useful comments. The present work was supported by JSPS KAKENHI Grant Number JP17H01148.
\end{acknowledgements}


\end{document}